\documentclass[conference,letterpaper,onethmnum]{IEEEtran}

\usepackage[utf8]{inputenc} 
\usepackage[T1]{fontenc}
\usepackage{url}
\usepackage{subfig}
\usepackage{ifthen}
\usepackage{cite}
\usepackage[cmex10]{amsmath} % Use the [cmex10] option to ensure complicance
\usepackage{booktabs}
\usepackage{bbm,bm}
\usepackage{enumerate}
\usepackage{enumitem}
\usepackage{graphicx} % Required for inserting images
\usepackage{amssymb}
\usepackage{amsthm}
\usepackage{xcolor}
\usepackage{todonotes}
\usepackage[hidelinks]{hyperref}
\usepackage{cleveref}

\theoremstyle{plain}
\newtheorem{theorem}{Theorem}

\newtheorem{corollary}[theorem]{Corollary}
\newtheorem{definition}{Definition}
\newtheorem{remark}[theorem]{Remark}

\crefname{corollary}{Corollary}{Corollaries}
\Crefname{corollary}{Corollary}{Corollaries}

\newcommand{\rank}{\mathrm{rank}}
\newcommand{\ttrank}{\mathrm{TT\textnormal{-}rank}}
\newcommand{\BL}{\bm{M}}

\newcommand{\BC}{\bm{C}}
\newcommand{\BR}{\bm{R}}
\newcommand{\BW}{\bm{W}}
\newcommand{\BM}{\bm{M}}
\newcommand{\BU}{\bm{U}}
\newcommand{\BV}{\bm{V}}

\newcommand{\BI}{\mathbb{I}}

\newcommand{\R}{\mathbb{R}}

\newcommand{\BSigma}{\bm{\Sigma}}
\newcommand{\ang}[1]{{\langle #1 \rangle}}

\let\oldsqrt\sqrt % rename \sqrt as \oldsqrt
\def\sqrt{\mathpalette\DHLhksqrt} % define the new \sqrt in terms of the old one
\def\DHLhksqrt#1#2{%
\setbox0=\hbox{$#1\oldsqrt{#2\,}$}\dimen0=\ht0
\advance\dimen0-0.2\ht0
\setbox2=\hbox{\vrule height\ht0 depth -\dimen0}%
{\box0\lower0.4pt\box2}}

\begin{document}
\title{Property Inheritance for Subtensors in \\Tensor Train Decompositions} 
%"Property Inheritance in Tensor Train Subtensors: Incoherence, Condition Numbers, and Rank Preservation"
%"Tensor Train Subtensors: A Study of Property Inheritance and Computational Efficiency"
%"Incoherence and Condition Number Inheritance in Tensor Train Dimensionality Reduction"
% %%% Single author, or several authors with the same affiliation:
% \author{%
%  \IEEEauthorblockN{Author 1 and Author 2}
% \IEEEauthorblockA{Department of Statistics and Data Science\\
%                    University 1\\
 %                   City 1\\
  %                  Email: author1@university1.edu}% }

%%% Several authors with up to three affiliations:
\author{%
  \IEEEauthorblockN{HanQin Cai}
  \IEEEauthorblockA{Department of Statistics and Data Science \\ Department of Computer Science \\
                    University of Central Florida\\
                    Orlando 32816, USA\\
                    Email: {hqcai@ucf.edu}%\href{mailto:hqcai@ucf.edu}{hqcai@ucf.edu}
                    }
  \and
  \IEEEauthorblockN{Longxiu Huang}
  \IEEEauthorblockA{Department of Computational Mathematics, Science and Engineering \\
  Department of Mathematics\\
                   Michigan State University\\
                   East Lansing 48824, USA\\
                    Email: {huangl3@msu.edu}%\href{mailto:huangl3@msu.edu}{huangl3@msu.edu}
                    }
}

\maketitle

%%%%%%
%% Abstract: 
%% If your paper is eligible for the student paper award, please add
%% the comment "THIS PAPER IS ELIGIBLE FOR THE STUDENT PAPER
%% AWARD." as a first line in the abstract. 
%% For the final version of the accepted paper, please do not forget
%% to remove this comment!
%%

\begin{abstract}
Tensor dimensionality reduction is one of the fundamental tools for modern data science. To address the high computational overhead, fiber-wise sampled subtensors that preserve the original tensor rank are often used in designing efficient and scalable tensor dimensionality reduction. However, the theory of property inheritance for subtensors is still underdevelopment, that is, how the essential properties of the original tensor will be passed to its subtensors. This paper theoretically studies the property inheritance of the two key tensor properties, namely incoherence and condition number, under the tensor train setting. We also show how tensor train rank is preserved through fiber-wise sampling. The key parameters introduced in theorems are numerically evaluated under various settings. The results show that the properties of interest can be well preserved to the subtensors formed via fiber-wise sampling. Overall, this paper provides several handy analytic tools for developing efficient tensor analysis methods.
\end{abstract}

\section{Introduction}
The analytic study for tensor data, i.e., multi-dimensional array of numbers, has received much attention since last decade. Among many others, dimensionality reduction is one of the most popular methodologies for tensor data analysis, that is, to find a low-rank expression of the given tensor.
Unlike the standard matrix rank definition, various ranks have been proposed for tensors; for instance, CP rank \cite{carroll1970analysis}, Tucker rank \cite{de2000multilinear}, tubal rank \cite{kilmer2011factorization}, and tensor train rank \cite{oseledets2009breaking}. 
 The tensor dimensionality reductions have found a wide range of applications, such as signal processing \cite{cichocki2015tensor,sidiropoulos2017tensor,chen2024Laplacian,chen2024correlating,wang2024TensorDynSamp}, computer vision \cite{vasilescu2002multilinear,yan2006multilinear,cai2021rtcur,tan2023tubal,chen2024coseparable}, social networks  \cite{kassab2021detecting,cai2024rtcur_journal,kassab2024sparseness}, and bioinformatics \cite{yener2008multiway,omberg2009global,cong2015tensor,hore2016tensor}. 

However, one major challenge for tensor dimensionality reduction is the high computational complexity associated with the complication of the high-dimensional structure. Recently, many Nyström-style methods have been developed for tensor and matrix dimensionality reductions under various tensor rank settings  {\cite{lidiak2022quantum,chen2022tensor,cai2021mode,hamm2020perspectives}}. These methods can significantly reduce the computational complexity of the dimensionality reductions via constructing and utilizing appropriated subtensors and submatrices which are potentially compact-sized. It is clear that good subtensors and submatrices are key to the success of Nyström-style dimensionality reductions.
While many studies have focused on the approximation theory aspect of the tensor Nyström-style methods, such as sampling and perturbation 
analysis  \cite{ghahremani2024cross,cai2021mode}, the theory of \textit{property inheritance}, i.e., how the essential tensor properties such as incoherence and condition number, will be passed to subtensors, is still underdeveloped. 

To address the vacancy in tensor theory, this paper studies the property inheritance for certain types of subtensors under the tensor train (TT) rank setting \cite{oseledets2009breaking}. Following the common settings, we assume the original tensor is exactly low rank, and we focus on the subtensors that keep the original TT rank of the tensor. With the sequential nature of tensor train decomposition, we present results to explain how the ranks are sequentially determined in subtensors (see \Cref{thm:rank-pre_1,thm:rank-pre_i}.) The subtensor property inheritance for incoherence and condition number is then presented in \Cref{thm:mu_of_R,thm:mu_of_C}. In \Cref{sec:numerical}, we empirically evaluate the values of key parameters introduced in \Cref{thm:mu_of_R,thm:mu_of_C}. The numerical results show that the tensor properties are well preserved with a simple sampling method. These results develop a deeper understanding of the properties of subtensors and can benefit future studies in related fields. 
%\todo[inline]{add numerical contribution}

\section{Notation and Preliminaries}
We start with some basic notation. Distinct typefaces are used for different numerical structures. Specifically, calligraphic capital letters (e.g., $\mathcal{T}$) represent tensors, boldface capital letters (e.g., $\BL$) denote matrices, regular capital letters (e.g., $I$) denote index sets, boldface lowercase letters (e.g., $\bm{v}$) are used for vectors, regular lower case letters (e.g., $\alpha$) indicate scalars. The set of the first \( d \) natural numbers is denoted by \([d] := \{1, \dots, d\} \).  

For a tensor $\mathcal{T}$, the notation $\mathcal{T}(I, :, \cdots, :)$ refers to slicing or extracting a subset of the tensor where the indices in the first mode are restricted to $I$, while all indices in the other modes are selected. For a matrix $\BM \in \mathbb{R}^{n_1 \times n_2}$, we use the following notations:  \begin{itemize}
     \item $\BM(I, :)$ denotes the $|I| \times n_2$ row submatrix of $\BM$ consisting only of the rows indexed by $I \subseteq [n_1]$;  \item $\BM(:, J)$ denotes the $n_1 \times |J|$ column submatrix of $\BM$ consisting only of the columns indexed by $J \subseteq [n_2]$;  \item $\BM(I, J)$ represents the $|I| \times |J|$ submatrix containing the entries $a_{ij}$ of $\BM$ for which $(i, j) \in I \times J$. 
 \end{itemize}
The \( n \times n \) identity matrix is denoted as \( \BI_n \). We reserve the letters $\BW_{\BM}$ and $\BV_{\BM}$ to denote the left and right singular vectors of a matrix $\BM$. Finally, $\BM^\dagger$ denotes the Moore-Penrose pseudoinverse of $\BM$.

%\HQ{For simplicity, we use the notation \( I_1 \otimes I_2 \) to denote the linear indices obtained from the Cartesian product of \( I_1 \) and \( I_2 \). }
%\todo{Add notation for subtensor/submatrix, Kronecker product}

For the matrices, the two essential properties of interest, incoherence and condition number, are defined as follows.
\begin{definition}[Matrix incoherence and condition number]
Let $\BL\in\R^{n_1\times n_2}$ be a rank-$r$ matrix, and let $\BL={\BW}_{\BL}\bm{\BSigma}_{\BL}{\BV}_{\BL}^\top$ be its compact SVD.  Then $\BL$ is said $\left\{\mu_{1,\BL},\mu_{2,\BL}\right\}$-incoherent (i.e., $\mu_{1,\BL}$-column-incoherent and $\mu_{2,\BL}$-row-incoherent) for some constants $\mu_{1,\BL}$ and $\mu_{2,\BL}$ such that
\begin{equation*}\label{EQN:A1}
   \begin{aligned} \left\| {\BW}_{\BL}\right\|_{2,\infty}\leq&\sqrt{\frac{\mu_{1,\BL} r}{n_1}} \quad\textnormal{and}\quad \left\|\bm{{\BV}}_{\BL}\right\|_{2,\infty}\leq &\sqrt{\frac{\mu_{2,\BL} r}{n_2}}. 
   \end{aligned} 
\end{equation*}
The condition number $\kappa_{\BL}$ is defined as $$\kappa_{\BL}:= \frac{\sigma_{1,\BL}}{\sigma_{r,\BL}},$$  where $\sigma_{i,\BL}$ is the $i$-th largest singular value  of $\BL$.
\end{definition}

In the authors' prior work \cite{cai2021robust}, property inheritance for submatrices has been thoroughly studied, i.e., how the incoherence
and condition number of a given matrix will transfer to its row and column
submatrices. The results can be summarized as the following theorem. 

\begin{theorem}[Section~3 of \cite{cai2021robust}] \label{THM:BetaBound}
Suppose $\BL\in\R^{n_1\times n_2}$ is rank-$r$ and $\{\mu_{1,\BL},\mu_{2,\BL}\}$-incoherent. Choose index set $I\subseteq[n_1]$ such that the row submatrix $\BR=\BL(I,:)$ is also rank-$r$. Then it holds
\begin{align*}
\mu_{2,\BR}&\leq\mu_{2,\BL}, \\
\mu_{1,\BR}&\leq\alpha^2\kappa_{\BL}^2\mu_{1,\BL},\\
\kappa_{\BR}&\leq\alpha\sqrt{\mu_{1,\BL} r}\kappa_{\BL},
\end{align*}
where $\alpha:=\sqrt{\frac{|I|}{n_1}}\left\|{\BW}_{\BL}(I,:)^\dagger\right\|_2$. 
Similarly, choose index set $J\subseteq[n_2]$ such that the column submatrix $\BC=\BL(:,J)$ is also rank-$r$. Then it holds
\begin{align*}
\mu_{1,\BC}&\leq\mu_{1,\BL}, \\
\mu_{2,\BC}&\leq\beta^2\kappa_{\BL}^2\mu_{2,\BL},\\
\kappa_{\BC}&\leq\beta\sqrt{\mu_{2,\BL} r}\kappa_{\BL},
\end{align*}
where $\beta:=\sqrt{\frac{|J|}{n_2}}\left\|{\BV}_{\BL}(J,:)^\dagger\right\|_2$.
% \begin{gather*}
% \mu_{1,\BC}\leq\mu_{1,\BL}, \quad \mu_{2,\BC}\leq\beta^2\kappa_{\BL}^2\mu_{2,\BL},\\
% \textnormal{and } \kappa_{\BC}\leq\beta\sqrt{\mu_{2,\BL} r}\kappa_{\BL}.
% \end{gather*}
%Similar results hold for rank-$r$ row submatrices. 
\end{theorem}

These results are handy tools for matrix analysis that involves submatrices. Naturally, researchers want to extend it to tensor settings. In fact, some recent work has successfully extended \Cref{THM:BetaBound} to tensors under tubal setting \cite{salut2023tensor,su2024guaranteed}. However, the study we proposed in this paper, i.e., extension to tensors in the TT setting, is rather complicated since the subtensors are obtained through sequential operations, just like TT decomposition itself. 

Next, we introduce some basic preliminaries for tensor and tensor-train decomposition. For a thorough introduction, we refer the reader to \cite{kolda2009tensor,oseledets2009breaking}.

% \begin{definition} [Kronecker product of matrices] The Kronecker product of $\BA\in\mathbb{R}^{n_1\times n_2}$ and $\BB\in\mathbb{R}^{m_1\times m_2}$ is denoted by $\BA\kron \BB$. The result is a matrix of size $(n_1m_1)\times (n_2m_2)$ defined by
% \begin{eqnarray*}
% \BA\kron \BB&=&\begin{bmatrix}
% \BA_{11}\BB&\BA_{12}\BB&\cdots &\BA_{1n_2}\BB\\
% \BA_{21}\BB&\BA_{22}\BB&\cdots &\BA_{2n_2}\BB\\
% \vdots&\vdots&\ddots&\vdots\\
% \BA_{n_11}\BB&\BA_{n_12}\BB&\cdots&\BA_{n_1n_2}\BB
% \end{bmatrix}.
% \end{eqnarray*}
% \end{definition}
\begin{definition}[Mode-$k$ product] Let $\mathcal{T}\in\mathbb{R}^{n_1\times n_2\times \cdots\times n_d}$ and $\BM\in\mathbb{R}^{J\times n_k}$, the multiplication between $\mathcal{T}$ on its $k$-th mode with $\BM$ is denoted as $\mathcal{X}=\mathcal{T}\times_k \BM$ with 
\begin{align*}
   &\mathcal{X}(i_1,\cdots,i_{k-1},j,i_{k+1},\cdots,i_{d})\\
   :=&\sum_{s=1}^{n_k}\mathcal{T}(i_1,\cdots,i_{k-1},s,i_{k+1},\cdots,i_{d})\BM(j,s). 
\end{align*}

% Note this can be written as a matrix product by noting that $\mathcal{X}_{(k)}=\mathbb{M} \mathcal{T}_{(k)}$. If we have multiple tensor matrix product from different modes, we  use the notation $\mathcal{X}\times_{i=t}^{s} A_i$ to denote the product $\mathcal{X}\times_{t}A_{t}\times_{t+1}\cdots\times_{s}A_{s}$.
\end{definition}
\begin{definition}[TT rank and $i$-th tensor unfolding]
   Let $\mathcal{T} \in \mathbb{R}^{n_1 \times n_2 \times \cdots \times n_d}$. The TT-rank of $\mathcal{T}$ is defined as $(r_1, r_2, \cdots, r_{d-1})$, where 
$r_i := \rank(\mathcal{T}_{\ang{i}})$, and the $i$-th tensor unfolding $\mathcal{T}_{\ang{i}} \in \mathbb{R}^{(n_1 \cdots n_i) \times (n_{i+1} \cdots n_d)}$ is obtained as:
\[
\mathcal{T}_{\langle i\rangle}(j_1 \cdots j_i, j_{i+1} \cdots j_d) = \mathcal{T}(j_1, j_2, \cdots, j_d).
\]
\end{definition}

Note the $i$-th tensor unfolding defined here is different from the mode-$i$ tensor unfolding.

\begin{definition}[TT decomposition] \label{def:TTD}
The TT decomposition of a tensor $\mathcal{T} \in \mathbb{R}^{n_1 \times n_2 \times \cdots \times n_d}$ with TT-rank $(r_1, r_2, \cdots, r_{d-1})$ is expressed as:
\[
\mathcal{T} := \mathcal{T}_1  \bullet \mathcal{T}_2\bullet   \cdots \bullet \mathcal{T}_d,
\]
where $\mathcal{T}_i \in \mathbb{R}^{r_{i-1} \times n_i \times r_i}$ are the core tensors, and $r_0 = r_d = 1$. Specifically, for $j_i \in [n_i]$, an entry of $\mathcal{T}$ can be written as:
\[
\mathcal{T}(j_1, j_2, \cdots, j_d) = \mathcal{T}_1(:, j_1, :) \mathcal{T}_2(:, j_2, :) \cdots \mathcal{T}_d(:, j_d, :).
\]
% \begin{figure}[t]
%     \centering
%     \includegraphics[width=\linewidth]{tensor-train.pdf}
%     \caption{\cite{yin2021towards}. Visual representation of a tensor train (TT) decomposition using tensor diagrams: each edge corresponds to a dimension of the underlying tensor. }
%     \label{fig:TT-TD}
% \end{figure}
% \begin{figure} 
%     \centering
%     \includegraphics[width=\linewidth]{TensorTrain.png}
%     \caption{Illustration of the Tensor Train (TT) format with TT rank $(r_1, r_2, \dots, r_{d-1}, 1)$. Each entry of the tensor is represented as the product of $d$ matrices, where the $k$-th matrix in the ``train" is selected based on the value of $j_k$.}
%     \label{fig:TT}
% \end{figure}
\end{definition}
%\tikzset{every picture/.style={line width=0.75pt}} %set default line width to 0.75pt        

To help the reader better understand, \Cref{fig:TT-TD} is included here to illustrate TT decomposition.

\begin{figure}[t]
    \centering
    \includegraphics[width=\linewidth]{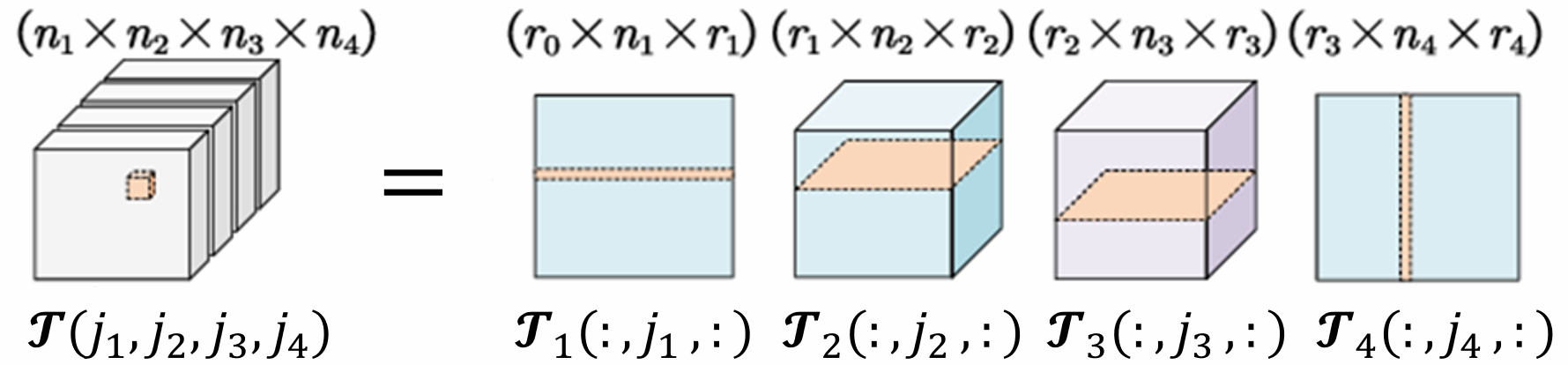}
    \caption{\cite{yin2021towards}. Visual representation of tensor train (TT) decomposition for a 4-order tensor. Note that $r_0=r_4=1$.}
    \label{fig:TT-TD}
\end{figure}

\begin{definition}[TT incoherence]
 Let $\mathcal{T}\in\mathbb{R}^{n_1\times n_2\times\cdots\times n_d}$ with $\ttrank(\mathcal{T})=(r_1,r_2,\cdots,r_{d-1})$. Then $\mathcal{T}$ is said $\{\bm{\mu}_{1,\mathcal{T}}, \bm{\mu}_{2,\mathcal{T}}\}$-incoherent, where $\bm{\mu}_{1,\mathcal{T}}, \bm{\mu}_{2,\mathcal{T}}\in\mathbb{R}^{d-1}$ and
\[
\bm{\mu}_{1,\mathcal{T}}(i) := \mu_{1,\mathcal{T}_{\ang{i}}}, \quad \bm{\mu}_{2,\mathcal{T}}(i) := \mu_{2,\mathcal{T}_{\ang{i}}}, \quad \forall i \in[d-1]. 
\]
% The incoherence of the TT-format tensor $\mathcal{T}\in\mathbb{R}^{n_1\times n_2\times\cdots\times n_d}$ is defined as
%  \[
%  (\{\mu_1(\mathcal{T}_{[1]},\mu_2(\mathcal{T}_{[1]}\},\{\mu_1(\mathcal{T}_{[2]},\mu_2(\mathcal{T}_{[2]}\},\cdots,\{\mu_1(\mathcal{T}_{[d-1]},\mu_2(\mathcal{T}_{[d-1]}\})
%  \]
\end{definition}

% \todo[inline]{We didn't use this definition.}
 
% For the matrix incoherence and condition number, \cite{cai2021robust} have the following results regarding how the incoherence and condition number of a given matrix will transfer to its submatrices:

% \begin{enumerate}[leftmargin=1cm]
%     \item $\max_i\left\| \bm {\BW}_{\BC}^T\bm{e}_i\right\|_2\leq \sqrt{\frac{\mu_1(\BL) r}{m}}$,\\
%     \item $\max_i\left\| {\BV}_{\BC}^T\bm{e}_i\right\|_2\leq \beta\kappa(\BL)\sqrt{\frac{\mu_2(\BL) r}{| J|}}$,\\  
%     \item $\kappa(\BC)\leq\beta\sqrt{\mu_2(\BL) r}\kappa(\BL)$.
% \end{enumerate}
% In particular, \[\mu_1(\BC)\leq\mu_1(\BL),\quad \textnormal{and}\quad \mu_2(\BC)\leq\beta^2\kappa(\BL)^2\mu_2(\BL).\]
%\end{theorem}

% \todo[inline]{state matrix to tensor extension is not easy}

\section{Main Results}
This section presents the main theoretical results for subtensor property inheritance under the tensor train (TT) setting.  

Under the matrix setting, property inheritance is studied for row/column-wise sampled submatrices with the same rank as the original matrix \cite{cai2021robust}. The reason is simple: if a submatrix has the rank of the original matrix, then it spans the same linear subspace, thus preserving the subspace information of the original matrix. 
Similarly, this paper aims at the subtensors with the same TT rank as the original tensor. 
As shown in \Cref{def:TTD}, TT decomposition is computed as a series of operations on tensor dimensions in a sequential order. 
This contrasts with some other tensor decompositions, such as CP and Tucker, whose operations are in no particular order of dimensions. 
Hence, to study the subtensors of interest, we must understand how fiber-wise sampling in the earlier dimension may impact the rank in later dimensions in TT decomposition.  

As a base case, this study finds that if we fiber-wise samples along the first dimension to form a subtensor such that that rank of the mode-1 unfolding, i.e., $r_1$, is preserved, then this subtensor will also preserve the rest of TT rank.
%, i.e., $\{r_j\}_{j=2}^{d-1}$. 
This result is presented as \Cref{thm:rank-pre_1}.

\begin{theorem}\label{thm:rank-pre_1}
  Let $\mathcal{T}\in\mathbb{R}^{n_1\times n_2\times\cdots\times n_d}$ with $\ttrank(\mathcal{T})=(r_1,r_2,\cdots,r_{d-1})$. Suppose that the index set $I\subseteq[n_1]$ is chosen such that $\rank(\mathcal{T}_{\ang{1}}(I,:))=r_1$. Then $\ttrank(\mathcal{T}(I,:,\cdots,:))=(r_1,r_2,\cdots,r_{d-1})$.  
\end{theorem}
\begin{proof}
Set subtensor \(\mathcal{R} = \mathcal{T}(I, :,\cdots,:)\).  
Since \(\rank(\mathcal{T}_{\ang{1}}(I,:)) = r_1\), i.e., \(\rank(\mathcal{R}_{\ang{1}}) = r_1\), by the theory of tensor CUR decomposition \cite[Theorem~2,3]{cai2021mode}, there exists \(J \subseteq [n_2 n_3 \cdots n_d]\) such that
\[
\mathcal{T} = \mathcal{R} \times_1 \BC \BU^\dagger,
\]
where \(\BC = \mathcal{T}_{\ang{1}}(:, J)\) and \(\BU = \mathcal{T}_{\ang{1}}(I, J)\).
Thus, we have:
\begin{equation}\label{eqn:CUR-expansion}
    \mathcal{T}_{\ang{i}} = (\BC\BU^\dagger \otimes \BI_{d_2 \cdots d_i}) \mathcal{R}_{\ang{i}}.
\end{equation}

From \eqref{eqn:CUR-expansion}, it follows that
\begin{equation}\label{eqn:r_originalllsub}
    \rank(\mathcal{T}_{\ang{i}}) \leq \rank(\mathcal{R}_{\ang{i}}) =r_i.
\end{equation}
Since \(\mathcal{R}_{\ang{i}}\) is a submatrix of \(\mathcal{T}_{\ang{i}}\), we also have
\begin{equation}\label{eqn:r_originalggsub}
    \rank(\mathcal{T}_{\ang{i}}) \geq \rank(\mathcal{R}_{\ang{i}})=r_i.
\end{equation}
Combining \eqref{eqn:r_originalllsub} and \eqref{eqn:r_originalggsub}, we conclude that 
\[
\rank(\mathcal{T}_{\ang{i}}) = \rank(\mathcal{R}_{\ang{i}}) = r_i.
\]
for all $i\in[d-1]$. This finishes the proof. 
%Therefore, the TT-rank of \(\mathcal{T}(I, :, \cdots, :)\) is \((r_1, r_2, \cdots, r_{d-1})\).
\end{proof}

Naturally, we can extend this result to the subtensors formed by fiber-wise sampling along the $i$-th tensor unfolding, i.e., sample along the first dimension of the $i$-th tensor unfolding of the original tensor, and preserve the subsequent part of TT rank $\{r_j\}_{j=i}^{d-1}$. This result is presented as \Cref{thm:rank-pre_i}.

\begin{corollary} \label{thm:rank-pre_i}
Let $\mathcal{T}= \mathcal{T}_1  \bullet \mathcal{T}_2\bullet   \cdots \bullet \mathcal{T}_d\in\mathbb{R}^{n_1\times n_2\times\cdots\times n_d}$ with $\ttrank(\mathcal{T})=(r_1,r_2,\cdots,r_{d-1})$. 
Given a fixed $i\in[d-1]$, suppose that the index set $I\subseteq \bigotimes_{j=1}^{i}[n_j]$ is chosen such that $\rank(\mathcal{T}_{\ang{i}}(I,:))=r_i$. Set $\mathcal{R}=(\mathcal{T}_1\bullet\cdots\bullet\mathcal{T}_i)_{\ang{i}}(I_i,:)\bullet\mathcal{T}_{i+1}\bullet\cdots\bullet\mathcal{T}_{d}$. Then $\ttrank(\mathcal{R})=(r_i,r_{i+1},\cdots,r_{d-1})$.  
\end{corollary}
\begin{proof}
    This result is a direct combination of the definition of TT decomposition and \Cref{thm:rank-pre_1}.
\end{proof}
Note that \Cref{thm:rank-pre_i} does not account for potential fiber-wise sampling that occurs in the dimensions preceding the \( i \)-th dimension. Due to the sequential nature of TT decomposition, preserving all TT ranks in the formed subtensors, namely \( \mathcal{R}_i \), requires that the index pool of fibers available for sampling in the later dimensions is constrained by the indices sampled in the preceding dimensions. Specifically, 
\begin{equation} \label{eq:I_i}
\begin{cases}
        I_0=\{1\},  \\
        I_i\subseteq I_{i-1}\otimes [n_i]\subseteq\bigotimes_{j=1}^i[n_j]
\end{cases}
\end{equation}
for all $i\in[d-1]$, where $I_i$ denotes the index set sampled in $i$-th tensor unfolding. 

In \Cref{thm:mu_of_R}, we present the property inheritance, in terms of TT incoherence and condition number, for the subtensors $\mathcal{R}_i$ formed with index sets $I_i$, given that the subtensors preserve the original TT rank.

\begin{theorem} \label{thm:mu_of_R}
Let $\mathcal{T}= \mathcal{T}_1  \bullet \mathcal{T}_2\bullet   \cdots \bullet \mathcal{T}_d\in\mathbb{R}^{n_1\times n_2\times\cdots\times n_d}$ with $\ttrank(\mathcal{T})=(r_1,r_2,\cdots,r_{d-1})$ and $\{\bm{\mu}_{1,\tau},\bm{\mu}_{2,\tau}\}$-incoherence. With index sets $I_i$ defined as \eqref{eq:I_i}, 
set 
\begin{align*}
&\mathcal{R}_{i}=(\mathcal{T}_1\bullet\cdots\bullet\mathcal{T}_i)_{\ang{i}}(I_i,:)\bullet\mathcal{T}_{i+1}\bullet\cdots\bullet\mathcal{T}_{d}
\end{align*}
for $i\in[d-1]$. Note that $\mathcal{R}_{i}\in\mathbb{R}^{|I_i|\times n_{i+1}\times\cdots\times n_d}$. Suppose that $I_i$ are chosen such that $\rank((\mathcal{R}_{i})_{\ang{1}})=r_i$ for all $i$. 
Then   $\{\bm{\mu}_{1,\mathcal{R}},\bm{\mu}_{2,\mathcal{R}}\}$ satisfies the following conditions:  for all $i\in[d-1]$ and $t\in[d-i]$, 
\begin{align*}
    \mu_{1,(\mathcal{R}_i)_{\ang{t}}}&\leq \alpha_{i,t}^2\kappa^2_{\mathcal{T}_{\ang{t+i-1}}}\mu_{1,\mathcal{T}_{\ang{t+i-1}}},\\ \mu_{2,(\mathcal{R}_{i})_{\ang{t}}} &\leq \mu_{2,\mathcal{T}_{\ang{t+i-1}}},\\
    \kappa_{(\mathcal{R}_i)_{\ang{t}}} &\leq \alpha_{i,t}\sqrt{\mu_{1,\mathcal{T}_{\ang{t+i-1}}}r_{t+i-1}}\kappa_{\mathcal{T}_{\ang{t+i-1}}}, 
\end{align*}
where 
 \begin{equation}
    \label{eqn:alpha_it}
%\begin{gather}
\alpha_{i,t}:=
\sqrt{\frac{|I_{i}|}{\prod_{j=1}^{i}n_j}}\left\|\BW_{\mathcal{T}_{\ang{t+i-1}}}\left(I_i\otimes\left(\bigotimes_{j=i+1}^{t+i-1}[n_j]\right),:\right)^\dagger\right\|_2.
%\end{gather}
 \end{equation}
\end{theorem}
 
\begin{proof}
Firstly, let's fix a $i\in[d-1]$. 
Since the index set $I_i$ gives $\rank((\mathcal{R}_{i})_\ang{1})=r_i$, by \Cref{thm:rank-pre_i}, we have that $\rank((\mathcal{R}_i)_\ang{t})=\rank(\mathcal{T}_\ang{t+i-1})$ for $t\in[d-i]$. 
We observed that $(\mathcal{R}_{i})_{\ang{t}}$ for a fixed $t\in[d-i]$ is a row submatrix of $\mathcal{T}_{\ang{t+i-1}}$ with row index $I_i\otimes (\bigotimes_{j=i+1}^{t+i-1}[n_j])$. Applying \Cref{THM:BetaBound} and set $$\alpha_{i,t}=\sqrt{\frac{|I_{i}|}{\prod_{j=1}^{i}n_j}}\left\|\BW_{\mathcal{T}_{\ang{t+i-1}}}\left(I_i\otimes\left(\bigotimes_{j=i+1}^{t+i-1}[n_j]\right),:\right)^\dagger\right\|_2,$$ we have that
\begin{align*}
    \mu_{1,(\mathcal{R}_i)_{\ang{t}}}&\leq\alpha_{i,t}^2\kappa^2_{\mathcal{T}_{\ang{t+i-1}}}\mu_{1,\mathcal{T}_{\ang{t+i-1}}},\\ \mu_{2,(\mathcal{R}_{i})_{\ang{t}}}&\leq\mu_{2,\mathcal{T}_{\ang{t+i-1}}},\\
    \kappa_{(\mathcal{R}_i)_{\ang{t}}}&\leq \alpha_{i,t} \sqrt{\mu_{1,\mathcal{T}_{\ang{t+i-1}}}r_{t+i-1}}\kappa_{\mathcal{T}_{\ang{t+i-1}}}.
\end{align*}
The above argument applies to any $i\in[d-1]$ and $t\in[d-i]$. 
Thus, it finishes the proof.
\end{proof}

Those subtensors can be viewed as generalized row submatrices of the unfoldings. Since the ``columns'' are original, we see no amplification on the ``column'' incoherence parameter $\mu_2$ while the ``row'' incoherence $\mu_1$ can be slightly amplified due to fiber-wise sampling. The condition number can also become slightly worse through the sampling.

Now that the ``row'' samplings along the first to $(d-1)$-st dimensions have been handled, we shift the focus to the properties of subtensors related to the ``column'' samplings. 
Similar to the matrix case where rows and columns can be sampled independently, the indices of generalized column samplings on the unfoldings, namely $J_i$, can be independent of $I_i$. Specifically, 
\begin{equation} \label{eq:J_i}
J_i\subseteq\left[\prod_{j=i+1}^{d}n_j\right]
\end{equation}
for $i$-th tensor unfolding. The property inheritance for the subtensors formed by both $I_i$ and $J_i$, namely $\BC_i=(\mathcal{R}_{i-1})_{\ang{1}}(:,J_i)$, is presented as \Cref{thm:mu_of_C}. 

\begin{theorem} \label{thm:mu_of_C}
Let $\mathcal{T}= \mathcal{T}_1  \bullet \mathcal{T}_2\bullet   \cdots \bullet \mathcal{T}_d\in\mathbb{R}^{n_1\times n_2\times\cdots\times n_d}$ with $\ttrank(\mathcal{T})=(r_1,r_2,\cdots,r_{d-1})$ and $\{\bm{\mu}_{1,\tau},\bm{\mu}_{2,\tau}\}$-incoherence. With index sets $I_i$ and $J_i$ defined as \eqref{eq:I_i} and \eqref{eq:J_i} respectively, 
set
\begin{align*}
\begin{cases}
&\mathcal{R}_0=\mathcal{T},\\ 
&\mathcal{R}_{i}=(\mathcal{T}_1\bullet\cdots\bullet\mathcal{T}_i)_{\ang{i}}(I_i,:)\bullet\mathcal{T}_{i+1}\bullet\cdots\bullet\mathcal{T}_{d},\\%\in\mathbb{R}^{|I_i|\times n_{i+1}\times\cdots\times n_d}
&\BC_i=(\mathcal{R}_{i-1})_{\ang{1}}(:,J_i)
\end{cases}
\end{align*}
for $i\in[d-1]$. Suppose that $I_i$ and $J_i$ is chosen such that $\rank((\mathcal{R}_{i})_{\ang{1}})=r_i$  and $\rank(\BC_i)=r_i$ for all $i$. 
Then it holds
\begin{align*}
\mu_{1,\BC_1}&\leq\mu_{1,\mathcal{T}_{\ang{1}}},\\   
\mu_{2,\BC_1}&\leq \beta_1^2\kappa^2_{\mathcal{T}_{\ang{1}}}\mu_{2,\mathcal{T}_{\ang{1}}}, \\
\kappa_{\BC_1}&\leq \beta_1\sqrt{\mu_{2,\mathcal{T}_{\ang{1}}}r_1}\kappa_{\mathcal{T}_{\ang{1}}}
\end{align*}
for $i=1$, and
\begin{align*}
\mu_{1,\BC_i}&\leq\alpha_i^2\beta_i^2\kappa^2_{\mathcal{T}_{\ang{i}}}r_i\mu_{1,\mathcal{T}_{\ang{i}}}\mu_{2,\mathcal{T}_{\ang{i}}}, \\
    \mu_{2,\BC_i}&\leq \beta_i^2\kappa^2_{\mathcal{T}_{\ang{i}}}\mu_{2,\mathcal{T}_{\ang{i}}}, \\
    \kappa_{\BC_i}
    &\leq \alpha_i\beta_i\sqrt{\mu_{1,\mathcal{T}_{\ang{i}}}\mu_{2,\mathcal{T}_{\ang{i}}}}r_i\kappa_{\mathcal{T}_{\ang{i}}},
\end{align*}
for $2\leq i\leq d-1$,
where 
\begin{equation}\label{eqn:alpha_beta}
\begin{aligned}
\alpha_i&:=\sqrt{\frac{|I_{i-1}|}{\prod_{j=1}^{i-1}n_j}}\left\|\BW_{\mathcal{T}_{\ang{i}}}(I_{i-1}\otimes[n_i],:)^\dagger\right\|_2,\\
\beta_i&:=\sqrt{\frac{|J_i|}{\prod_{j=i+1}^dn_{j}}}\left\|\BV_{\mathcal{T}_{\ang{i}}}(J_i,:)^\dagger\right\|_2.
\end{aligned}
\end{equation}
Note that $\alpha_1$ is not used in the theorem; however, we can take  $\alpha_1=1$ since $\mu_{1,\BC_1}\leq\mu_{1,\mathcal{T}_{\ang{1}}}$ with an implicit $\alpha_1$.
\end{theorem}
\begin{proof}
Firstly, let's consider the case of $i=1$. That is, $\BC_1=\mathcal{T}_{\ang{1}}(:,J_1)$. Since $J_1$ is chose such that $\rank(\BC_1)=r_1$,  by applying \Cref{THM:BetaBound}, we have that
\begin{align*}
    \mu_{1,\BC_1}&\leq\mu_{1,\mathcal{T}_{\ang{1}}},\\
    \mu_{2,\BC_1}&\leq \beta_1^2\kappa^2_{\mathcal{T}_{\ang{1}}}\mu_{2,\mathcal{T}_{\ang{1}}}, \\
    \kappa_{\BC_1}&\leq \beta_1\sqrt{\mu_{2,\mathcal{T}_{\ang{1}}}r_1}\kappa_{\mathcal{T}_{\ang{1}}},
\end{align*}
% \begin{gather*}
%     \mu_{1,\BC_1}\leq\mu_{1,\mathcal{T}_{\ang{1}}},
%     \quad
%     \mu_{2,\BC_1}\leq \beta_1^2\kappa^2_{\mathcal{T}_{\ang{1}}}\mu_{2,\mathcal{T}_{\ang{1}}}, \\
%     \kappa_{C_1}\leq \beta_1\sqrt{\mu_{2,\mathcal{T}_{\ang{1}}}r_1}\kappa_{\mathcal{T}_{\ang{1}}},
% \end{gather*}
where \begin{equation*}%\label{eqn:beta1}
    \beta_1:=\sqrt{\frac{|J_1|}{|\prod_{j=2}^{d}n_j|}}\left\|\BV_{\mathcal{T}_{\ang{1}}}(J_1,:)^\dagger\right\|_2.
\end{equation*} 
 
Next, we consider the cases of $2\leq i\leq d-1$. Notice that $\BC_i=(\mathcal{R}_{i-1})_{\ang{1}}(:,J_i)=\mathcal{T}_{\ang{i}}(I_{i-1}\otimes [n_i],J_i)$. Additionally, we are given $\rank(\BC_i)=r_i$ with the chosen $J_i$. 
By applying \Cref{THM:BetaBound}, we have that
\begin{align*}
\mu_{1,\mathcal{T}_{\ang{i}}(:,J_i)}&\leq\mu_{1,\mathcal{T}_{\ang{i}}},\\
\mu_{2,\mathcal{T}_{\ang{i}}(:,J_i))}&\leq \beta_i^2\kappa^2_{\mathcal{T}_{\ang{i}}}\mu_{2,\mathcal{T}_{\ang{i}}},\\
\kappa_{\mathcal{T}_{\ang{i}}(:,J_i)}&\leq\beta_i\sqrt{\mu_{2,\mathcal{T}_{\ang{i}}}r_i}\kappa_{\mathcal{T}_{\ang{i}}},
\end{align*}
where $\beta_i=\sqrt{\frac{|J_i|}{\prod_{j=i+1}^dn_{j}}}\left\|\BV_{\mathcal{T}_{\ang{i}}}(J_i,:)^\dagger\right\|_2$. 

Invoking \Cref{THM:BetaBound} again, we have that
\begin{align*}
    \mu_{1,\BC_i}&\leq\alpha_i^2\kappa^2(\mathcal{T}_{\ang{i}}(:,J_i))\mu_{1,\mathcal{T}_{\ang{i}}(:,J_i)} \\  &\leq\alpha_i^2\beta_i^2\kappa^2_{\mathcal{T}_{\ang{i}}}r_i\mu_{1,\mathcal{T}_{\ang{i}}}\mu_{2,\mathcal{T}_{\ang{i}}},
\end{align*}

\vspace{-0.2in}
\begin{align*}
    \mu_{2,\BC_i}&\leq \mu_{2,\mathcal{T}_{\ang{i}}(:,J_i)}\\
    &\leq\beta_i^2\kappa^2_{\mathcal{T}_{\ang{i}}}\mu_{2,\mathcal{T}_{\ang{i}}},
\end{align*}
and
\begin{align*}
    \kappa_{\BC_i}&\leq\alpha_i\sqrt{\mu_1(\mathcal{T}_{\ang{i}}(:,J_i))r_i}\kappa_{\mathcal{T}_{\ang{i}}(:,J_i)}\\
    %\leq&\alpha_i \sqrt{\mu_1(\mathcal{T}_{\ang{i}})r_i} \beta_i\sqrt{\mu_2(\mathcal{T}_{i})r_i}\kappa(\mathcal{T}_{\ang{i}})\\
    &\leq \alpha_i\beta_i\sqrt{\mu_1(\mathcal{T}_{\ang{i}})\mu_2(\mathcal{T}_{\ang{i}})}r_i\kappa(\mathcal{T}_{\ang{i}}),
\end{align*}
where $\alpha_i=\sqrt{\frac{|I_{i-1}|}{\prod_{j=1}^{i-1}n_j}}\left\|\BW_{\mathcal{T}_{\ang{i}}}(I_{i-1}\otimes[n_i],:)^\dagger\right\|_2$. This completes the proof.
\end{proof}

\begin{remark} \label{rmk:parameters}
    The bounds of property inheritance presented in \Cref{thm:mu_of_R,thm:mu_of_C} rely on the parameters $\alpha_{i,t}$, $\alpha_i$, and $\beta_i$ as defined in \eqref{eqn:alpha_it} and \eqref{eqn:alpha_beta}, which involvesthe Moore-Penrose pseudoinverses of subsampled left and right singular vectors of the unfoldings. The spectral bounds of these pseudoinverses highly depend on how the index sets $I_i$ and $J_i$ are sampled, and some sampling strategies may improve those parameters. This has been thoroughly discussed for matrix setting in \cite{cai2021robust}. However, due to page limits, we will only empirically verify those parameters (see \Cref{sec:numerical}) and reserve the theoretical discussion for future work.
\end{remark}

\section{Numerical Evaluation} \label{sec:numerical}
%\todo[inline]{Explain why only test $\beta$.}

\begin{figure*}[ht]
    \centering
\includegraphics[width=0.329\linewidth]{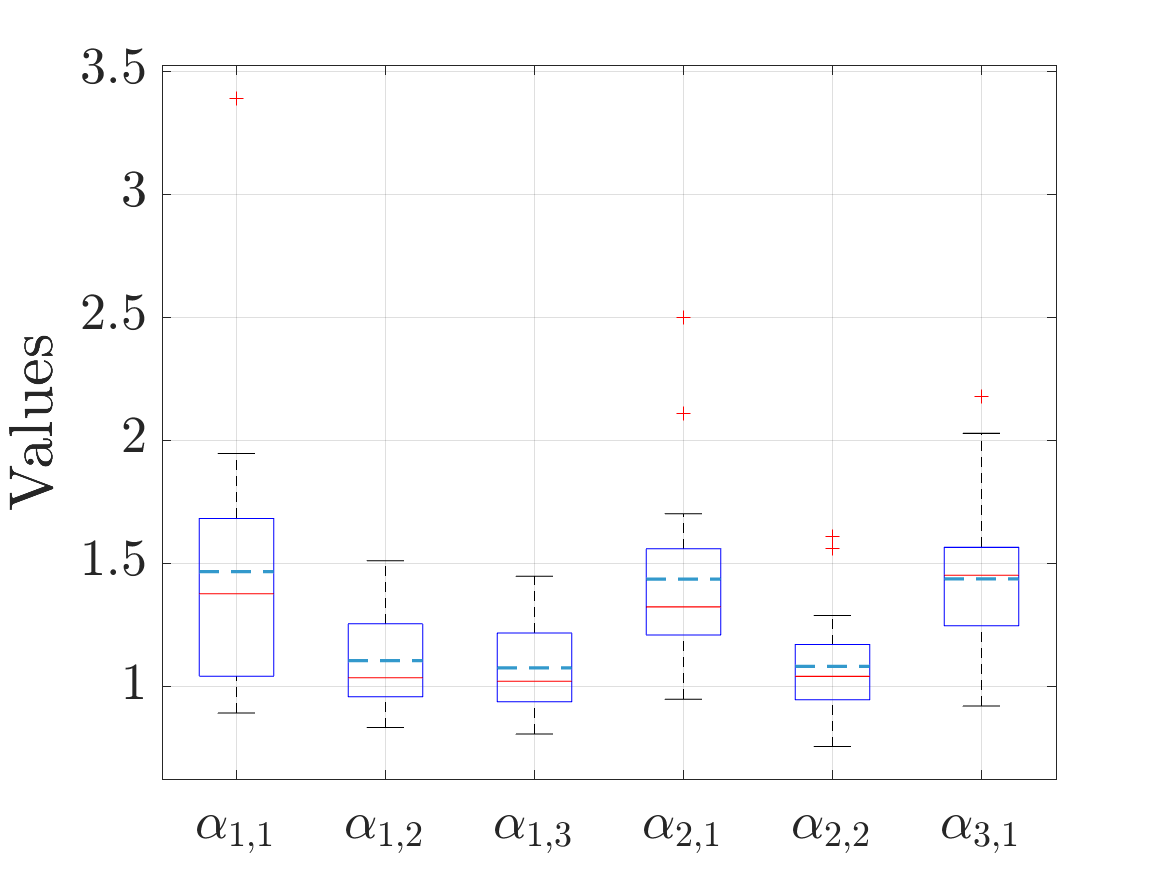}
\includegraphics[width=0.329\linewidth]{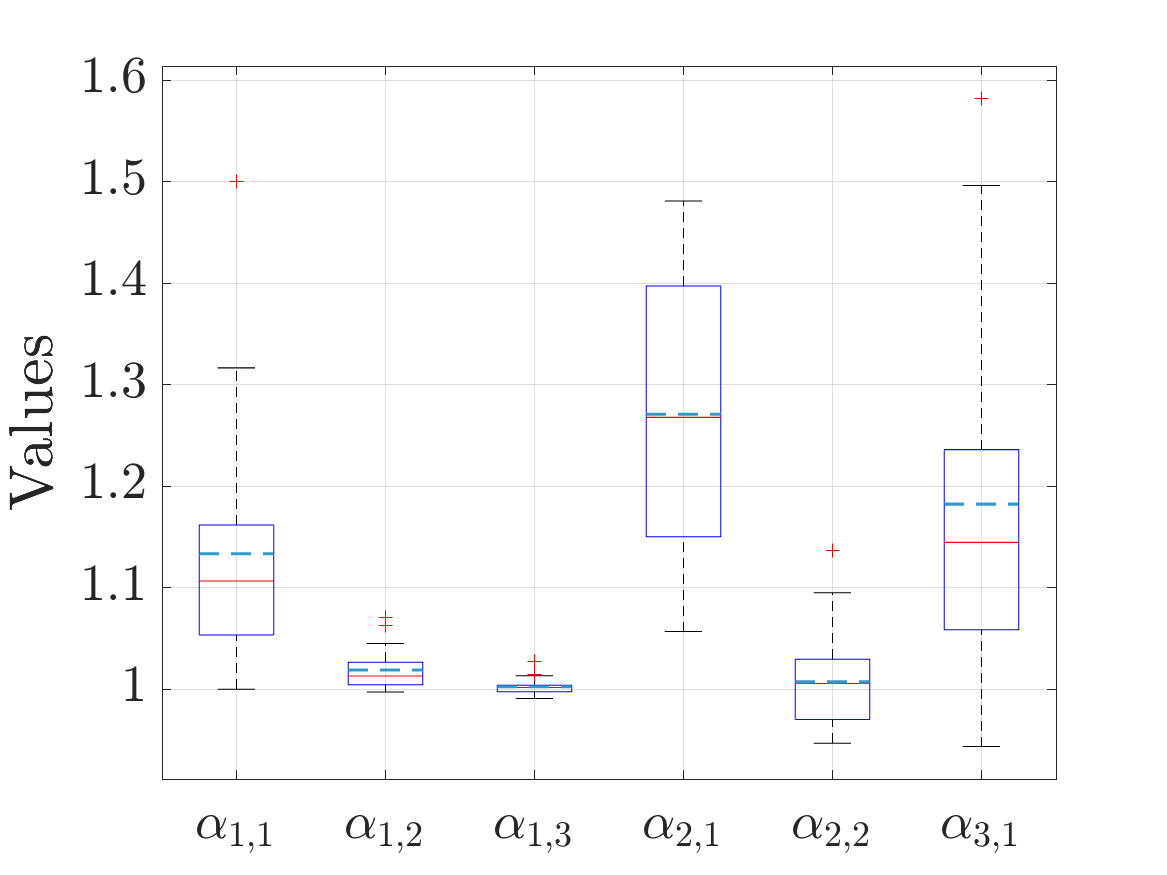}
\includegraphics[width=0.329\linewidth]{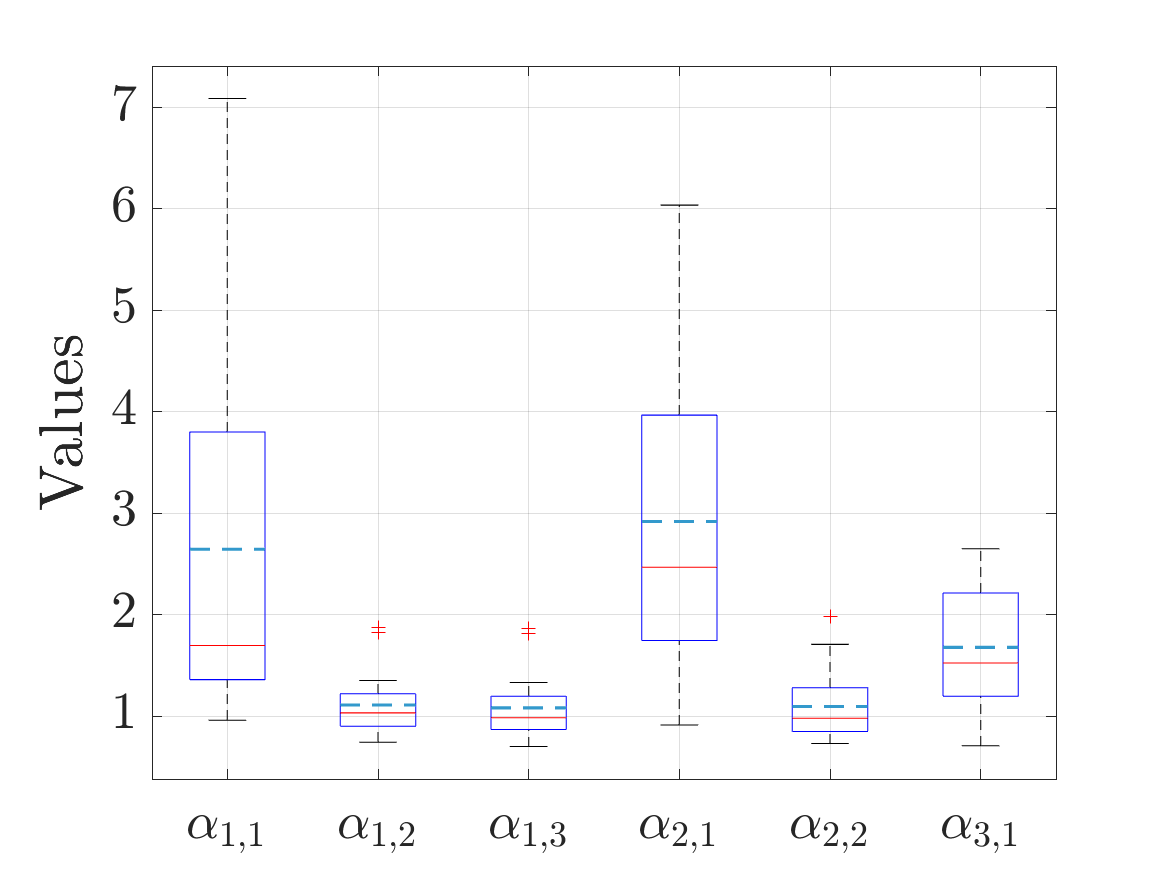}
    \caption{Boxplot for $\alpha_{i,t}$ as introduced in \Cref{thm:mu_of_R}. Each box represents the distribution of parameter values over 20 trials, showing the median (center line), interquartile range (box), and potential outliers ({red $+$}). The whiskers (top and bottom horizontal lines) extend to the most extreme data points within 1.5 times the interquartile range. The dashed blue line indicates the mean of the parameter values.  {\bf Left:} Gaussian generation; {\bf Middle:} Hadamard  generation; {\bf Right:} Uniform generation.}
    \label{fig:thmR}
\end{figure*}
\begin{figure*}[ht]
    \centering
\includegraphics[width=0.329\linewidth]{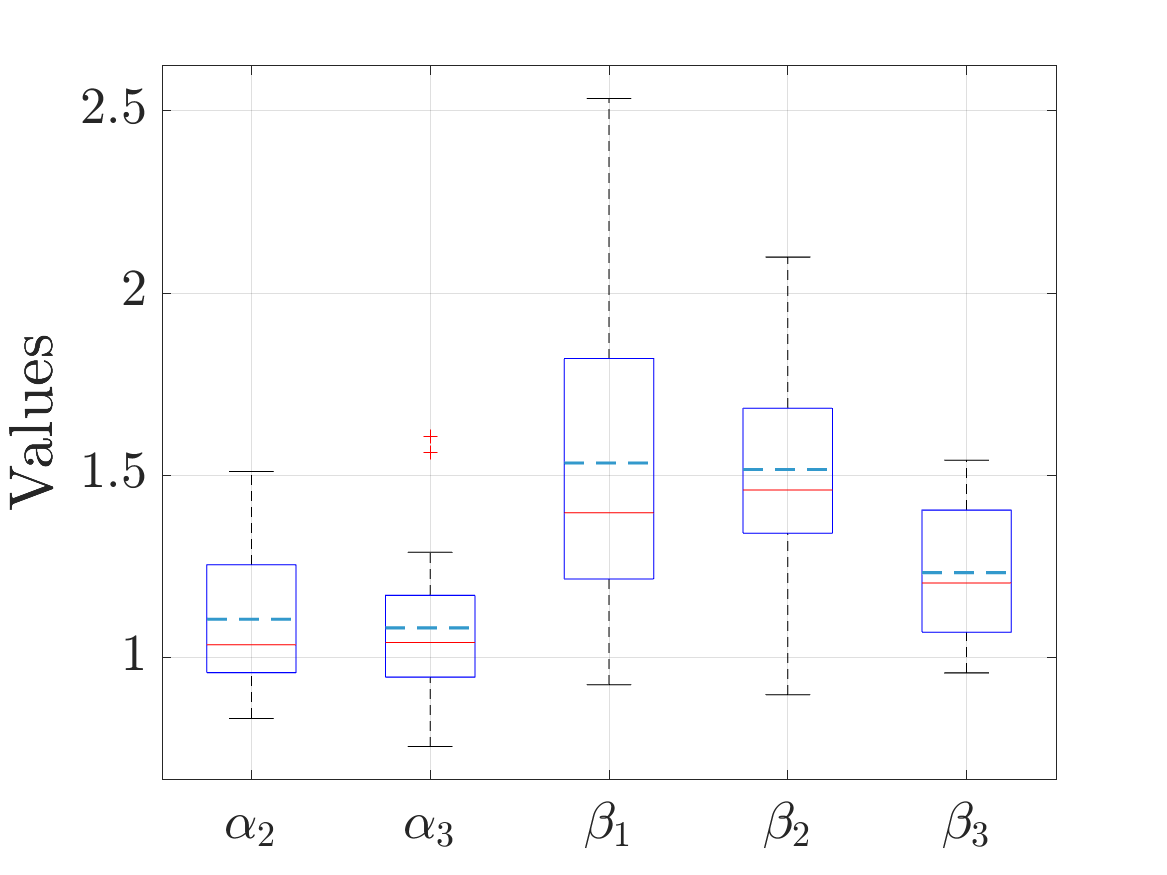}
\includegraphics[width=0.329\linewidth]{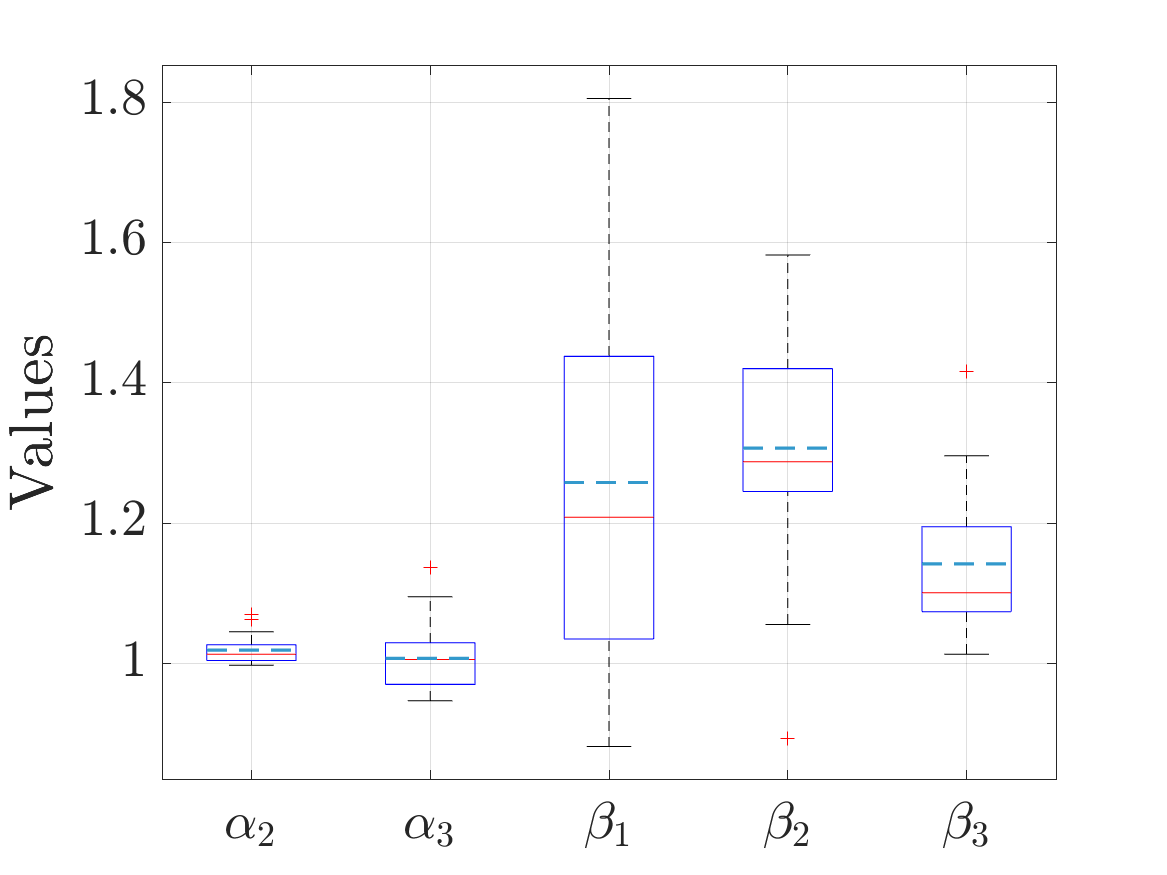}
\includegraphics[width=0.329\linewidth]{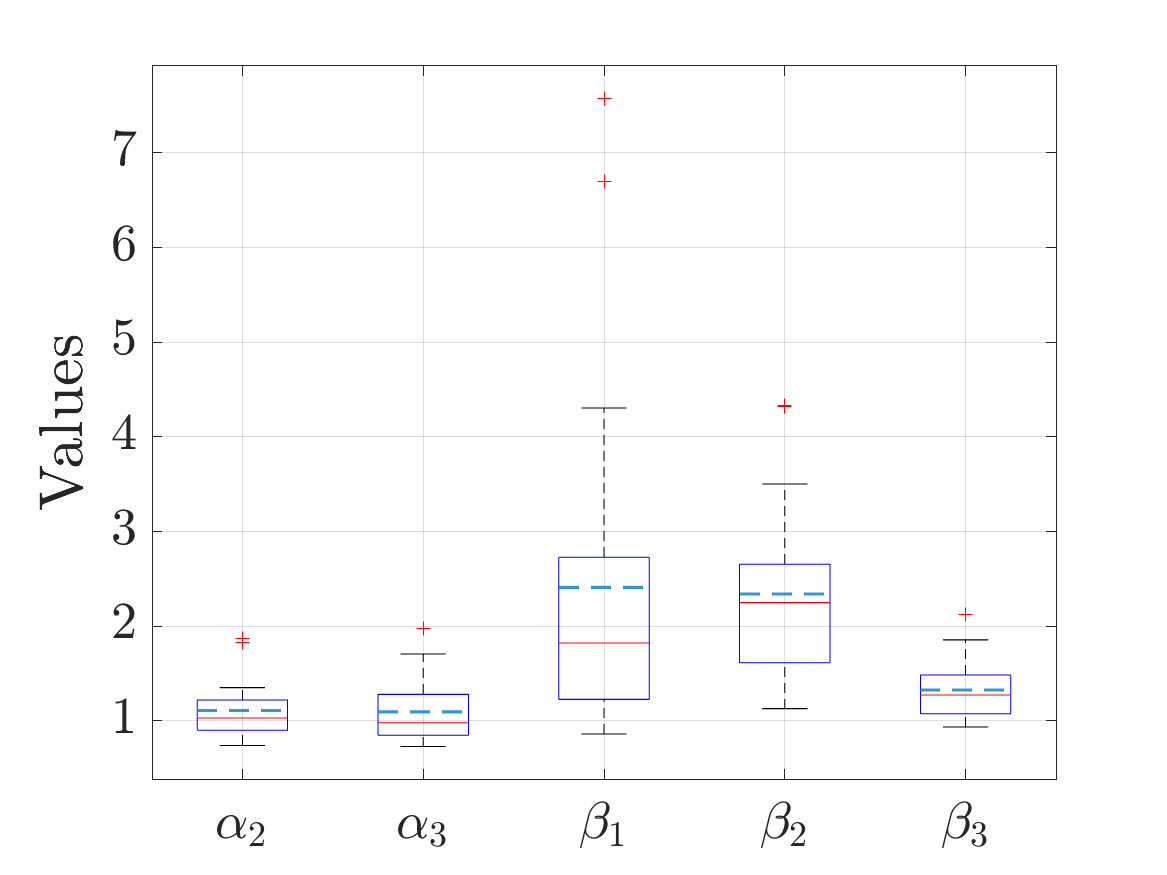}
    \caption{Boxplot for $\alpha_i$ and $\beta_i$ as introduced in \Cref{thm:mu_of_C}. The setup of the boxplot is the same as \Cref{fig:thmR}.  {\bf Left:} Gaussian generation; {\bf Middle:} Hadamard generation; {\bf Right:} Uniform generation.}
    \label{fig:thmC}
\end{figure*}

The values of the parameters \( \alpha_{i,t} \), \( \alpha_i \), and \( \beta_i\), as introduced in \Cref{thm:mu_of_C,thm:mu_of_R}, are key to controlling the essential properties of the subtensors. As discussed in \Cref{rmk:parameters}, they rely on the sampling method for the index sets. 
This section numerically evaluates these parameters with one of the most simple yet popular methods---\textit{uniform sampling without replacement}. 
All experiments are implemented on Matlab R2022b and executed on a laptop equipped with Intel i7-11800H CPU 
%(2.3GHz @ 8 cores) 
and 16GB DDR4 RAM.  

For test data, we generate four-dimensional tensors \( \mathcal{T} \in \mathbb{R}^{100 \times 100 \times 100 \times 100} \) with \( \text{TT-rank}(\mathcal{T}) = (r_1, r_2, r_3) = (2, 3, 2) \) using the following three random methods:
\begin{itemize}
    \item \textbf{Gaussian generation:} Set \( \mathcal{T} = \mathcal{T}_1 \bullet \mathcal{T}_2 \bullet \mathcal{T}_3 \bullet \mathcal{T}_4 \), where each entry of \( \mathcal{T}_i \in \mathbb{R}^{r_{i-1} \times 100 \times r_i} \) is independently sampled from a Gaussian distribution with mean 0 and variance 1.
    \item \textbf{Hadamard generation:} Set \( \mathcal{T} = \mathcal{T}_1 \bullet \mathcal{T}_2 \bullet \mathcal{T}_3 \bullet \mathcal{T}_4 \), where each entry of \( \mathcal{T}_i \in \{-1, 1\}^{r_{i-1} \times 100 \times r_i} \) is sampled independently with equal probability, i.e., 50\% chance for -1 and 50\% chance for 1.
        \item \textbf{Uniform  generation:} Set \( \mathcal{T} = \mathcal{T}_1 \bullet \mathcal{T}_2 \bullet \mathcal{T}_3 \bullet \mathcal{T}_4 \), where each entry of \( \mathcal{T}_i \in [0, 1]^{r_{i-1} \times 100 \times r_i} \) is sampled independently from a uniform distribution over the interval \([0, 1]\).
\end{itemize}
For every tensor \( \mathcal{T} \) generated using the above methods, we compute the left and right singular matrices \( \BW_{\mathcal{T}_\ang{i}} \) and \( \BV_{\mathcal{T}_\ang{i}} \) for each unfolding. The index sets \( I_i \subseteq I_{i-1} \otimes [100] \) with \( I_0 = \{1\} \) and \( J_i \subseteq [100^{4-i}] \) are sampled uniformly without replacement. The values of $\alpha_{i,t}$, $\alpha_i$ and $\beta_i$ are then calculated according to \eqref{eqn:alpha_it} and \eqref{eqn:alpha_beta}.
% \begin{align*}
%     \alpha_{i,t} :=& \sqrt{\frac{|I_i|}{100^i}} \left\| \BW_{\mathcal{T}_{\ang{t+i-1}}} \left(I_i \otimes [100^{t-1}], : \right)^\dagger \right\|_2,\\
% \alpha_i :=& \sqrt{\frac{|I_{i-1}|}{100^{i-1}}} \left\| \BW_{\mathcal{T}_{\ang{i}}}(I_{i-1} \otimes [100], :)^\dagger \right\|_2,\\
% \beta_i :=& \sqrt{\frac{|J_i|}{100^{4-i}}} \left\| \BV_{\mathcal{T}_{\ang{i}}}(J_i, :)^\dagger \right\|_2.
% \end{align*}

For each problem setup, we repeat the experiment 20 times. The results are reported as boxplots in \Cref{fig:thmR}, where the legend is detailed in the figure caption. 
%the dashed blue line represents the mean value of \( \alpha_{i,t} \) over 20 trials. 
Similarly, \Cref{fig:thmC} reports the boxplots for \( \alpha_i \) and \( \beta_i \).

As shown in \Cref{fig:thmR,fig:thmC}, the values of \( \alpha_{i,t} \), \( \alpha_i \), and \( \beta_i \) are relatively small with high confidence across all experiments. This empirical observation demonstrates that even with indices generated by simple uniform sampling without replacement, the subtensors well preserve the essential properties of the original tensor, i.e., incoherence and condition number, to some extent for all three random data generation methods.

\section{Conclusion and Future Work}
This paper presents a pilot study of property inheritance in subtensors formed by fiber-wise sampling under the tensor train (TT) setting. By focusing on subtensors that preserve the TT rank of the original tensor, we establish theoretical results that elucidate the inheritance of essential tensor properties such as incoherence and condition number. We numerically evaluate the values of key parameters from the theorems and show that the tensor properties are well preserved with simple uniform sampling without replacement. This paper provides a deeper understanding of the relationships between tensors and their subtensors, offering valuable analytic tools for advancing efficient and scalable tensor analysis. 

Future work will include a detailed theoretical discussion of the bounds for key parameters with different sampling strategies. 
%By developing tighter bounds and understanding their dependence on sampling strategies, we aim to further optimize tensor dimensionality reduction methods. 
Utilizing those properties of subtensors, we aim to further optimize tensor dimensionality reduction methods. 
Additionally, we plan to extend the results to other tensor decomposition frameworks. %These efforts will enhance the applicability of tensor analysis in tackling high-dimensional data challenges across diverse domains.

%%%%%%
%% Appendix:
%% If needed a single appendix is created by
%%
%\appendix
%%
%% If several appendices are needed, then the command
%%
% \appendices
%%
%% in combination with further \section commands can be used.
%%%%%%

\section*{Acknowledgment}
This work was partially supported by NSF DMS 2304489. The authors contributed equally.

\bibliographystyle{IEEEtran}
\bibliography{ref}

\end{document}